\documentclass[%
 reprint,
 amsmath,amssymb,
 aps,
]{revtex4-1}

\usepackage{graphicx}
\usepackage{dcolumn}
\usepackage{bm}
\usepackage{hyperref}

\usepackage[caption=false]{subfig}
\usepackage{siunitx}
\usepackage{amsthm}
\usepackage{physics}

\newtheorem{theorem}{Theorem}

\newtheorem{lemma}[theorem]{Lemma}

\DeclareMathOperator{\supp}{supp}

\begin{document}

\preprint{APS/123-QED}

\title{Photon Counting Distribution for Arrays of Single-Photon Detectors}

\author{Mattias J\"{o}nsson}
\email{matjon4@kth.se}
\author{Gunnar Bj\"{o}rk}%
\email{gbjork@kth.se}
\affiliation{%
 Department of Applied Physics, KTH Royal Institute of Technology\\
 AlbaNova University Center, SE 106 91 Stockholm, Sweden
}%




\date{\today}

\begin{abstract}
We derive a computationally efficient expression of the photon counting distribution for a uniformly illuminated array of single photon detectors. The expression takes the number of single detectors, their quantum efficiency, and their dark-count rate into account.  Using this distribution we compute the error of the array detector by comparing the output to that of a ideal detector. We conclude from the error analysis that the quantum efficiency must be very high in order for the detector to resolve a hand-full of photons with high probability. Furthermore, we conclude that in the worst-case scenario the required array size scales quadratically with the number of photons that should be resolved. We also simulate a temporal array and investigate how large the error is for different parameters and we compute optimal size of the array that yields the smallest error.
\end{abstract}

\maketitle


\section{\label{sec:introduction}INTRODUCTION}
Photon-number-resolving detectors have applications in various optical fields, such as investigation of exceptional points in $\mathcal{PT}$-symmetric systems \cite{QuirozJuarez2019}, measurements in the number basis \cite{Kovalenko2018}, quantum key exchange \cite{PhysRevA.98.012333}, photon-counting laser-radars \cite{Huang2014}, X-ray astronomy \cite{Holland1999}, evaluation of single-photon sources \cite{Hadfield2005}, and elementary-particle detection \cite{Haba2008}. These detectors can essentially be divided into two classes, inherent detectors and multiplexed detectors. The former case use some internal physical mechanism to deduce how many photons hit the detector, e.g., a transition edge sensor \cite{Irwin1995, Rosenberg2005a, Rosenberg2005b, Lita2008, Lita2009, Fukuda2011, Cahall2017}, while the latter case consists of detector arrays which use the combined outputs of multiple single-photon detectors \cite{Banaszek2003, Fitch2003, Achilles2003, Achilles2004, Eraerds2007, Jiang2007, Divochiy2008, Guerrieri2010, Afek2009, Natarajan2013, Mattioli2015, Nehra2017, Young2019}.

These detector arrays are easy to model physically, but the resulting probability distribution is computationally inefficient to evaluate for larger array sizes due to the high number of possible outcomes. This is problematic since the probability distribution can be used to, for example, improve of quantum key distribution \cite{Qi2019} and to evaluate the performance of temporally multiplexed arrays \cite{Kruse2017}.

In Ref. \cite{Sperling2012} an analytical expression for the probability distribution of a uniformly illuminated array is derived. The resulting expression is given as an expectation value of a normal ordered operator, which unfortunately makes the expression difficult to evaluate for a general situation.

In this paper we re-derive the photon counting distribution for a uniformly illuminated array consisting of single photon detectors. The distribution takes quantum efficiency, dark count probability and finite array size into account. The resulting formula is in agreement with photon counting formula derived in Ref. \cite{Sperling2012}, but it is expressed in a form that is simpler, and therefore numerically more efficient to evaluate.

Using the photon counting distribution we investigate the effects of quantum efficiency, dark count probability and array size on the detector performance. We also apply the distribution to investigate the performance of temporal arrays and show that there exists a finite array size for which the error is minimal for a given number of photons used as input. Our results are qualitatively in agreement with Ref. \cite{Kruse2017}, but our respective methods for analysis and assumptions differ.

\section{\label{sec:photon-counting-distribution}PHOTON COUNTING DISTRIBUTION}
Consider an array consisting of $n$ indistinguishable single photon detectors with quantum efficiency $\eta$ and dark count probability $p_d$. Assume that the single photon detectors click with a probability $1 - (1 - \eta)^m$ when $m$ photons hit the detector, that the probability per event for a dark count is $p_d$ and that the detector is memory-less. The probability for the single photon detector to click when $m$ photons hit that detector is then given by
\begin{equation}
    \Pr(k \mid m; \eta, p_d) = k + (1 - 2 k) (1 - \eta)^m (1 - p_d),
    \label{eq:single-photon-detector-probability}
\end{equation}
where $k \in \{ 0, 1 \}$ is the output corresponding to that the detector does not click or does click, respectively.

Let us now derive the photon counting distribution, which is an expression for the conditional probability to get $k$ clicks from the combined output of all the $n$ detectors in the array when $m$ photons where used as illumination. Let us encode the number of photons that hit each detector with the vector $\vec{x} \in \mathbb{N}^n$, where component $i$ corresponds to the number of photons that hit the detector at position $i$. Under uniform illumination, the photon counting distribution for the array is a multinomial distribution given by
\begin{equation}
\begin{split}
    &\Pr(k \mid m; \eta, p_d, n) = \sum_{\norm{\vec{x}}_1 = m} \frac{m!}{\prod_{i = 1}^n x_i!} \frac{1}{n^m} \binom{n}{k} \times\\
    &\times \Bigg( \prod_{j = 1}^k \Pr(1 \mid x_j; \eta, p_d) \Bigg) \Bigg( \prod_{l = k + 1}^n \Pr(0 \mid x_l; \eta, p_d) \Bigg).
\end{split}
\end{equation}

The photon counting distribution can be evaluated in its current form. However, this form requires a summation of up to $\order{2^{2n} \sqrt{n}}$ terms to get all non-zero probabilities for a specified $m$, which quickly becomes too large to evaluate exactly. It is therefore highly interesting to rewrite the distribution into a computationally effective form. To do this we utilize that the array is symmetric under any permutation of the single photon detectors, which together with the multinomial theorem gives (see theorem \ref{the:photon-counting-distribution} for details)
\begin{equation}
\begin{split}
    &\Pr(k \mid m; \eta, p_d, n) = \frac{1}{n^m} \binom{n}{k} \sum_{l = 0}^k (-1)^l \times\\
    &\times (1 - p_d)^{n - k + l} \binom{k}{l} \big[ n - (n - k + l) \eta \big]^m.
\end{split}
\end{equation}
This expression is easier to evaluate, since only $\order{n^2}$ terms needs to be computed in order to find all non-zero probabilities for a given $m$. Furthermore, this version of the distribution also makes it easier to analytically investigate the detector performance.

\section{\label{sec:errors-and-convergence}ERRORS AND CONVERGENCE}

\begin{figure}[t]
    \centering
    \includegraphics[width=\linewidth]{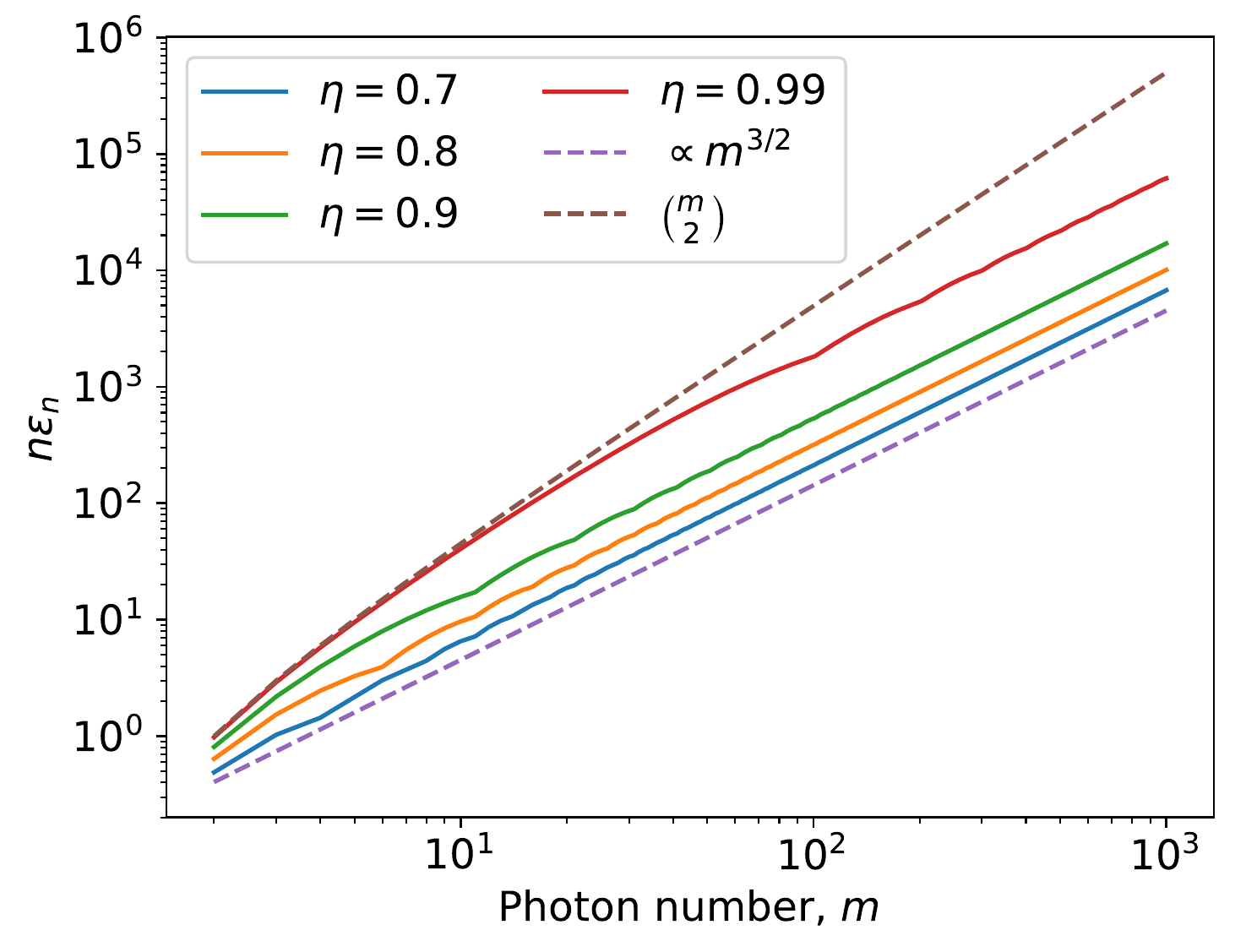}
    \caption{Scaling behavior for the finite size error $\eta_n$ when terms of order $\order{n^{-2}}$ are neglected. The error scales as between $m^{3/2}/n$ and $m^2/n$ depending on the quantum efficiency of the quantum efficiency of the system.}
    \label{fig:finite-size-scaling-for-large-n}
\end{figure}

Let us investigate the performance of a detector array by comparing how far the detector is from an ideal photon-number-resolving detector. We define an ideal photon-number-resolving detector to always gives an output equal to the number of incident photons, i.e. $\Pr_{\text{ideal}}(k \mid m) = \delta_{k, m}$. In our model such a detector corresponds to $\eta = 1$, $p_d = 0$ and that $n \to \infty$.

We define the error as the $L_1$ metric distance to an ideal detector
\begin{equation}
\begin{split}
    \epsilon &= \frac{1}{2} \norm{\Pr(\cdot \mid m; \eta, p_d, n) - \Pr(\cdot \mid m; 1, 0, \infty)}_1\\
    &= \frac{1}{2} \sum_{k \in \mathbb{N}} \abs{\Pr(k \mid m; \eta, p_d, n) - \Pr(k \mid m; 1, 0, \infty)}.
    \label{eq:total-error}
\end{split}
\end{equation}
The error $\epsilon \in [0, 1]$ and $\epsilon = 0$ if and only if the detector is ideal, whereas $\epsilon = 1$ corresponds to a detector characterized by $\eta$ and $n_d$ which is incapable of outputting the same value $m$ as the ideal detector for any $n$, i.e. $\supp(\Pr(\cdot \mid m; \eta, p_d, n)) \cap \supp(\Pr(\cdot \mid m; 1, 0, \infty)) = \varnothing$, where the support $\supp(f) = \{ x \mid f(x) \neq 0 \}$ is the set for which the function is non-zero.

To study the error $\epsilon$ it is possible to split it into three terms corresponding to the dark count error $\epsilon_d$, the quantum efficiency error $\epsilon_\eta$ and the finite size error $\epsilon_n$. By defining these errors appropriately we get with the triangle inequality that
\begin{equation}
\begin{split}
    \epsilon &\leq \frac{1}{2} \norm{\Pr(\cdot \mid m; \eta, p_d, n) - \Pr(\cdot \mid m; \eta, 0, n)}_1\\
    &+ \frac{1}{2} \norm{\Pr(\cdot \mid m; \eta, p_d, n) - \Pr(\cdot \mid m; \eta, 0, \infty)}_1\\
    &+ \frac{1}{2} \norm{\Pr(\cdot \mid m; \eta, 0, \infty) - \Pr(\cdot \mid m; 1, 0, \infty)}_1\\
    &= \epsilon_d + \epsilon_n + \epsilon_\eta,
\end{split}
\end{equation}
where the the last equality define the three errors.

The dark count error $\epsilon_d$ is in general difficult to evaluate, however it is straightforward to find an analytical expression for the case when the array is not illuminated. In this special case it holds that
\begin{equation}
    \epsilon_d = 1 - (1 - p_d)^n.
    \label{eq:dark-count-error}
\end{equation}
This result is of importance since it should intuitively corresponds to the largest dark count error possible for a given $n$ and since the result suggests that the maximal dark count error is negligible in most cases. For instance, superconducting nano-wire detectors may have a dark count probability around $p_d = \SI{1e-4}{}$, which would make the error $\epsilon_d$ less than $\SI{1}{\percent}$ for an array with $100$ detectors.

\begin{figure}[t]
    \centering
    \includegraphics[width=\linewidth]{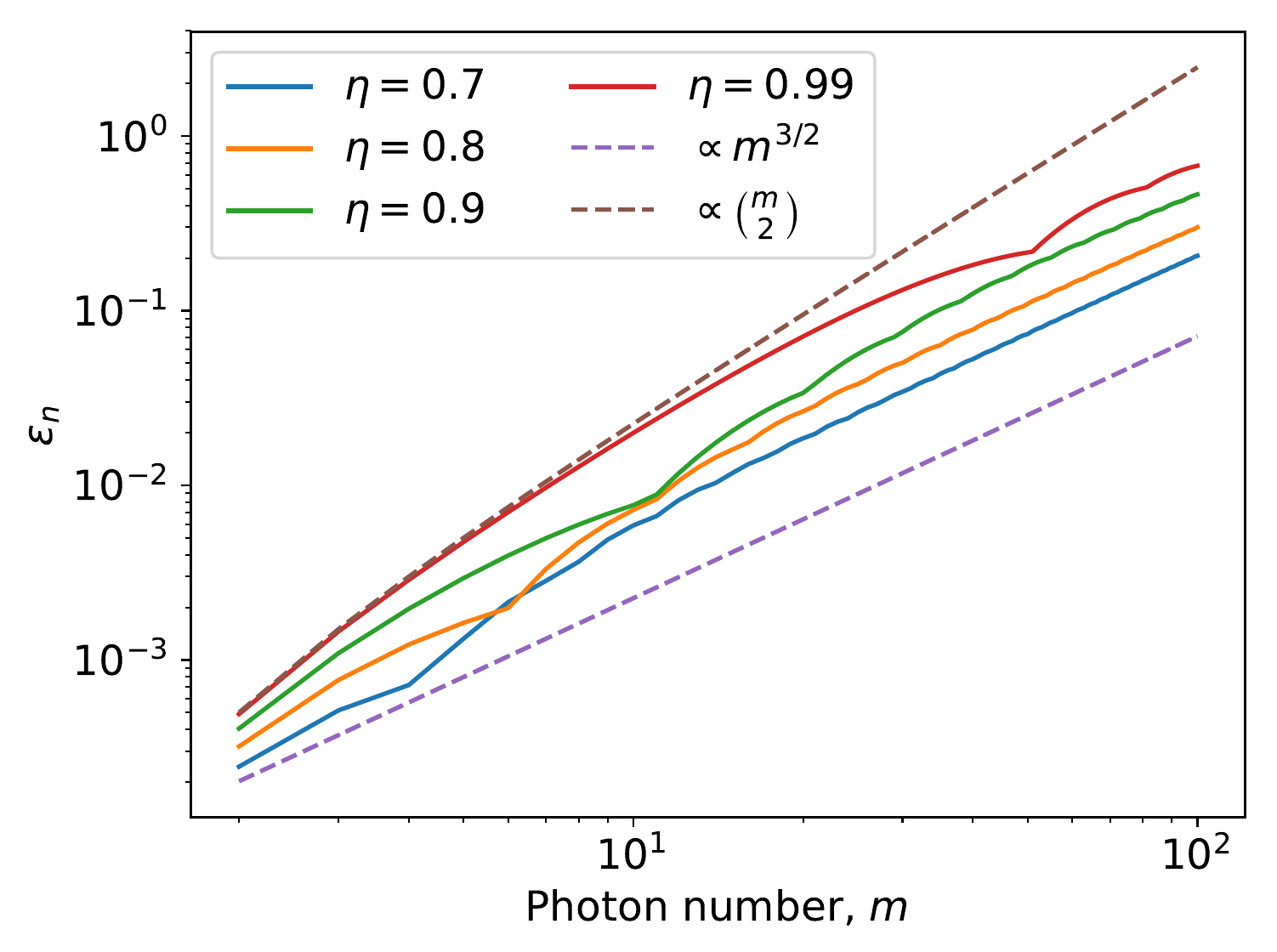}
    \caption{Exact scaling behavior for the finite size error $\eta_n$ when $n = 1000$. The error $\eta_n$ has the same behavior as was predicted when terms of $\order{n^{-2}}$ were neglected, which shows that the approximation did not alter the qualitative behavior.}
    \label{fig:finite-size-n-1000}
\end{figure}

The finite size error $\eta_n$ can be shown with simulations to scale between $m^{3/2}/n$ and $m^2/n$ for large enough $n$ (see Fig. \ref{fig:finite-size-scaling-for-large-n} and Fig. \ref{fig:finite-size-n-1000}). The exact scaling behavior depends on the quantum efficiency. When $\eta = 1$ it is possible to show that (see theorem \ref{the:finite-size-error} for details)
\begin{equation}
    \epsilon_n = \frac{1}{n} \binom{m}{2} + \order{n^{-2}} \sim \frac{m^2}{2n} + \order{n^{-2}},
\end{equation}
which shows that the number of single photon detectors required grows quadratically with the number of photons being resolved.

For $\eta < 1$ simulations suggest that the finite size scaling behavior is $m^{3/2}/n$ given that $m$ is large enough. However, for smaller $m$ simulations suggest that the scaling behavior may grow more rapidly, which is problematic since most practical photon-number-resolving detectors are prevented by the quantum-efficiency error from operating in the high $m$ regime.

In analogy with ( \ref{eq:dark-count-error}), the quantum efficiency error takes the form
\begin{equation}
    \epsilon_\eta = 1 - \eta^m.
\end{equation}
This shows that the requirement on quantum efficiency grows quickly with number of photons in the input. For instance, if the single photon detectors have $\eta = 0.9$ then the array can only detect $6$ photons with an error $\epsilon_\eta < 0.5$, which shows that the quantum efficiency is a major challenge in the creation of photon-number-resolving detector.

\section{\label{sec:temporal-array-simulation}TEMPORAL ARRAY SIMULATION}
\begin{figure}[t]
    \centering
    \includegraphics[width=0.7\linewidth]{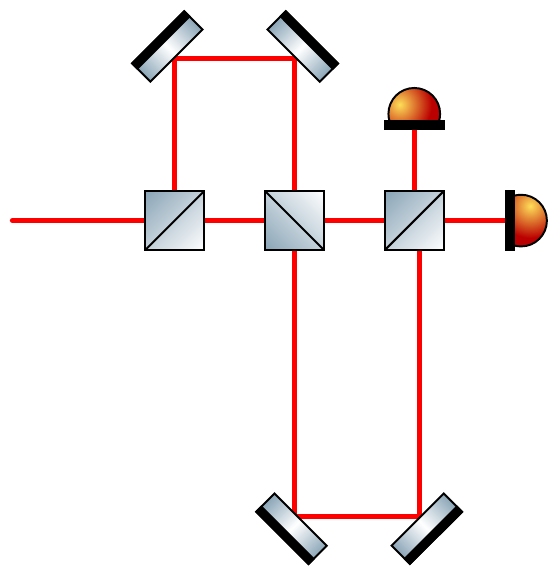}
    \caption{A schematic image of a temporal detection array for pulsed light. The input signal is divided into multiple time-bins by couplers. If the path lengths are chosen appropriately so the single photon detectors have time to recover between each pulse, then it is possible to use only two single-photon detectors.}
    \label{fig:temporal-array-schematics}
\end{figure}

\begin{figure}[t]
    \centering
    \includegraphics[width=\linewidth]{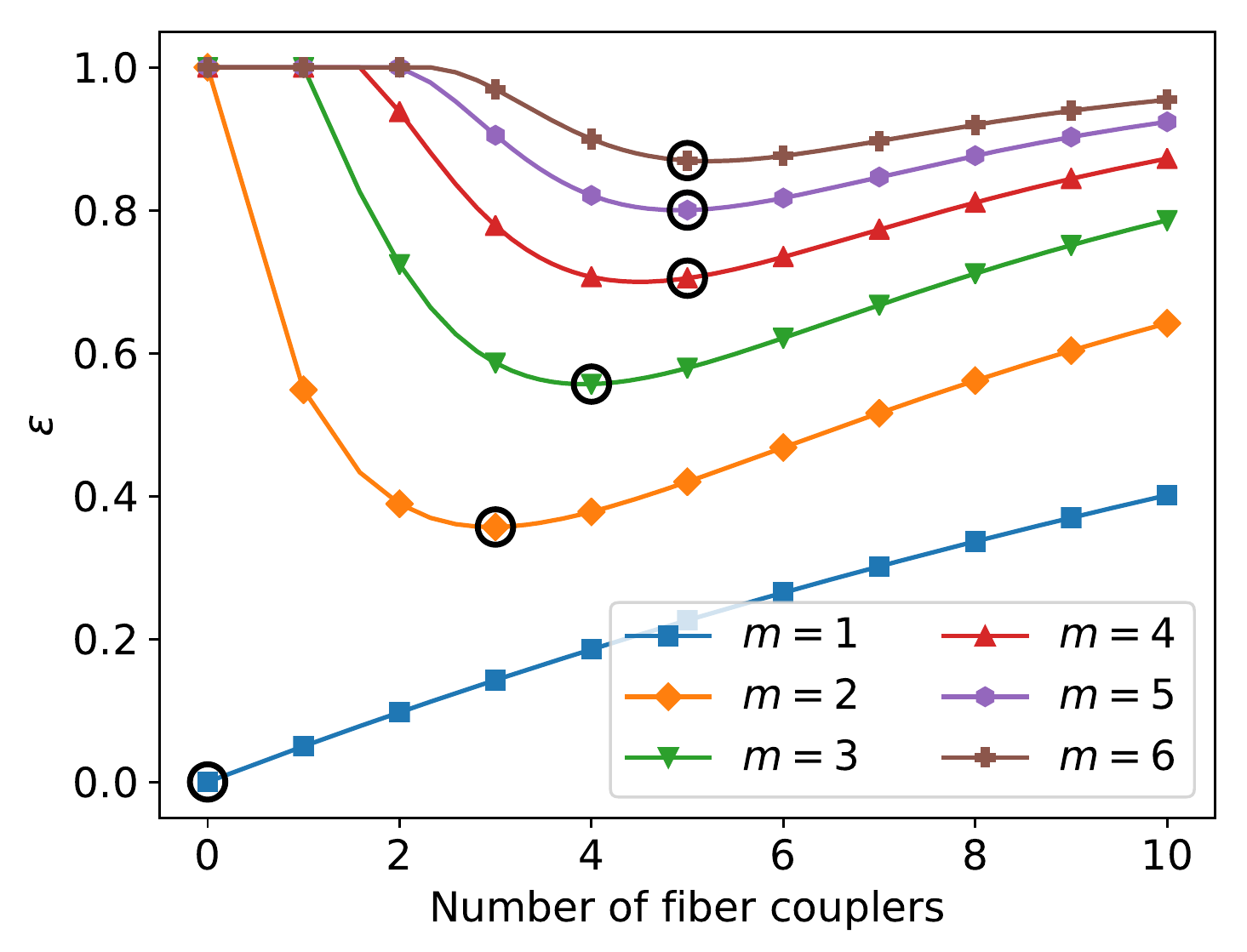}
    \caption{The total error $\epsilon$ as a function of the number of couplers when $\eta = 1.0$ and $\eta_c = 0.95$. The minimal errors for a given $m$ is marked with a black circle. As expected there exists a finite array size that minimizes the error which depends on the number of input photons and the loss in the couplers.}
    \label{fig:optimal-temporal-array-etac-095}
\end{figure}

\begin{figure}[t]
    \centering
    \includegraphics[width=\linewidth]{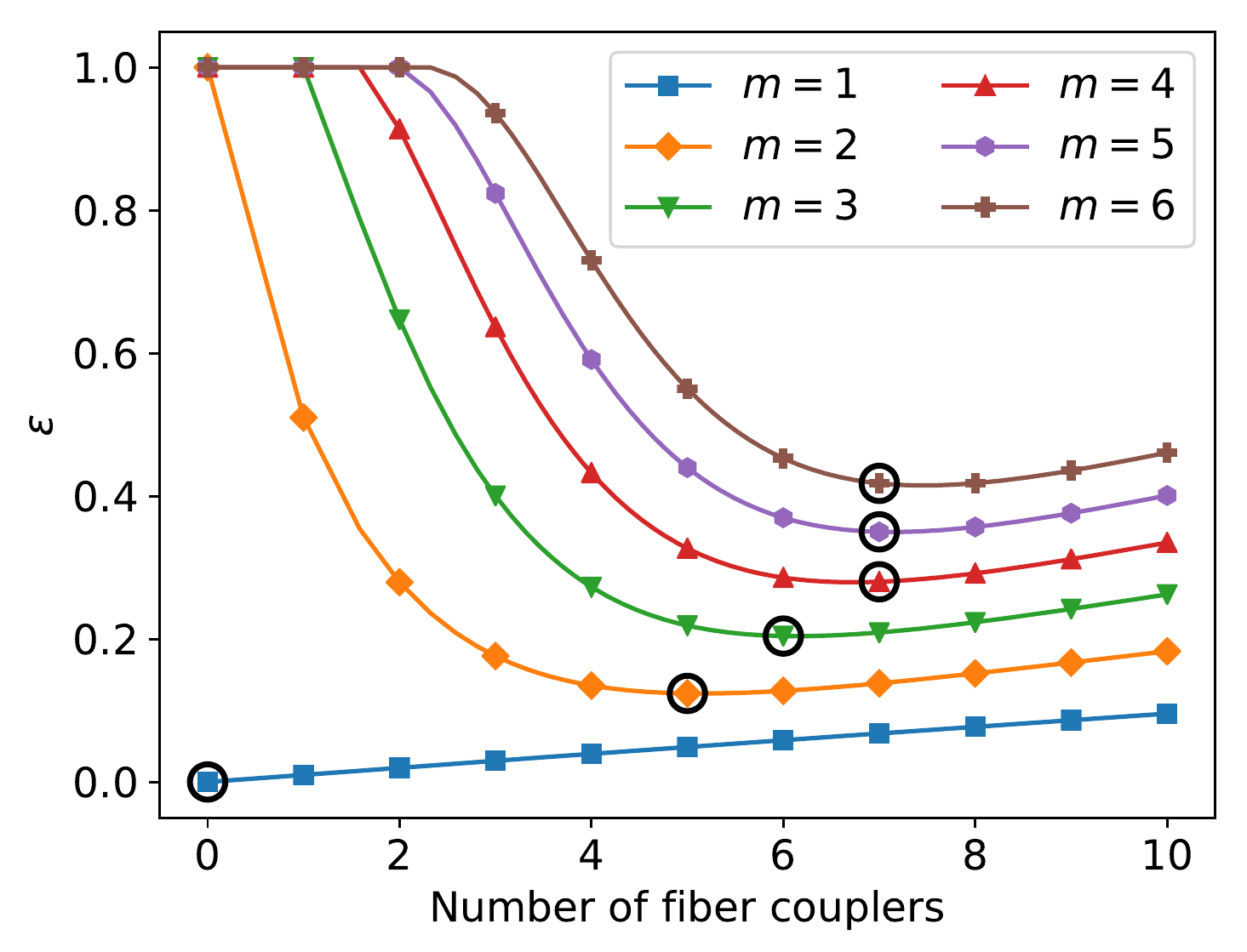}
    \caption{The total error $\epsilon$ as a function of the number of couplers when $\eta = 1.0$ and $\eta_c = 0.99$. The minimal errors for a given $m$ is marked with a black circle. Quite naturally, compared to Fig. \ref{fig:optimal-temporal-array-etac-095} the error is lower for higher quantum efficiencies and the optimal array size is generally larger.}
    \label{fig:optimal-temporal-array-etac-099}
\end{figure}

Ref. \cite{Fitch2003} suggested that the use of a temporal array (see Fig. \ref{fig:temporal-array-schematics}) could reduce the required number of single photon detectors in the array. However, this method has the drawback that it reduces temporal resolution. Furthermore, the effective quantum efficiency of the single photon detectors is reduced by the insertion of optical components. Hence, the total error $\epsilon$ takes its minimal value for a finite number of detectors \cite{Kruse2017}. The optimal array size will in general depend on the number of photons used as input and on the loss in each optical component.

To investigate optimal sizes for the arrays we model it as consisting of $n = 2^N$ dark-count-free single-photon detectors, where $N$ is the number of couplers. We assume that the couplers have linear losses and that a fraction $\eta_c$ of the photons survive to the next part of the circuit.  The effective quantum efficiency for each detector is then given by
\begin{equation}
    \eta_{\text{eff}} = \eta_c^N \eta,
\end{equation}
where $\eta$ is the quantum efficiency for the single photon detector.

In this model the optimal array size is independent of the detector quantum efficiency $\eta$ since the loss is linear and can therefore be modeled as an attenuation of the input signal. The effect of this loss in an increased error, although it does not change probability for two or more photons to hit the same single photon detector. It is therefore possible to run all simulations with $\eta = 1.0$. Hence, the total error for the array is a function of $N$ given by
\begin{equation}
    \epsilon = \frac{1}{2} \norm{\Pr(\cdot \mid m; \eta_c^N, 0, 2^N) - \Pr(\cdot \mid m; 1, 0, \infty)}_1.
\end{equation}

In Fig. \ref{fig:optimal-temporal-array-etac-095}, Fig. \ref{fig:optimal-temporal-array-etac-099} the result from  simulations where $\eta_c = 0.95$ and $\eta_c = 0.99$ are presented. As expected there exists an optimal array size which grows with the number of input photons and with higher $\eta_c$. Furthermore, as expected the overall error is lower in the latter figure since the quantum efficiency error is smaller.

The two simulations also show that the insertion loss plays a large role in the construction of a temporal array. Independent of $\eta$, when $\eta_c = 0.95$  the maximal number of photons that can be measured with $\epsilon < 0.5$ is two, while when $\eta_c = 0.99$ it is possible to measure up seven photons with $\eta < 0.5$ when everything else is ideal. This puts a strict requirement on the the optical components when building a temporal array.

\section{\label{sec:summary}SUMMARY}
In this paper we derive a closed analytical form of the probability distribution for a uniformly illuminated array of single photon detectors. We show that the full distribution can be computed in $\order{n^2}$ number of summations, which makes it substantially more efficient than by summing the contributions from all possible outcomes.

Investigations of the errors based on realistic numbers show that the dark count error is often small in comparison to the other errors if the array size is not too large. The finite size error grows in a worst case scenario as $\order{m^2/n}$, which implies that the number of single-photon detectors needs to be much larger than the square number of photons that the detector should resolve. The quantum efficiency error requires that the quantum efficiency needs to be very high, $>95$ \%, in order to resolve up to a hand-full of photons with high probability.

By building a temporal array it is possible to reduce the requirement on the number of single photon detectors to two. The drawback with this scheme is that the temporal resolution is reduced and that each added coupler reduce the effective quantum efficiency of the system. The consequence is that there exists a limit on the number of couplers that gives the optimal temporal array size for a given number of one wishes to photons.

\section*{\label{sec:acknowledgments}Acknowledgments}
This work was supported by the Knut and Alice Wallenberg Foundation grant "Quantum Sensing", the Swedish Research Council (VR) through grant 621-2014-5410, and through its support of the Linn\ae us Excellence Center ADOPT.

The simulations were performed on resources provided by the Swedish National Infrastructure for Computing (SNIC) at PDC.

Fig. \ref{fig:temporal-array-schematics} was created with ComponentLibrary by Alexander Franzen (\url{http://www.gwoptics.org/ComponentLibrary/}).

\appendix

\section{\label{sec:deriving-pcd}PHOTON COUNTING DISTRIBUTION}
\begin{lemma}\label{the:sum-invariance}
Let $m, \alpha, \beta \in \mathbb{R}$, $\vec{x} \in \mathbb{R}^n$ and $S \subseteq \{ l \in \mathbb{N} \mid l \leq n \}$. Introduce $S_o(j)$ to be the set containing the $j$ smallest elements of $S$, let $f: \mathbb{R}^n \to \mathbb{R}$ be a function which is invariant under any permutation of the input and let $g: \mathbb{R}^{\abs{S}} \to \mathbb{R}$. It then holds that
\begin{equation}
\begin{split}
    &\sum_{\norm{\vec{x}}_1 = m} f(\vec{x}) g(\{ x_l \mid l \not\in S \}) \prod_{j \in S} \big( 1 - \alpha \beta^{x_j} \big)=\\
    &\sum_{\norm{\vec{x}}_1 = m} f(\vec{x}) g(\{ x_l \mid l \not\in S \}) \sum_{j = 0}^{\abs{S}} (- \alpha)^j \binom{\abs{S}}{j} \prod_{k \in S_o(j)} \beta^{x_k}.
\end{split}
\end{equation}
\end{lemma}

\begin{proof}
Expanding the product yields that
\begin{equation}
\begin{split}
    &I = \sum_{\norm{\vec{x}}_1 = m} f(\vec{x}) g(\{ x_l \mid l \not\in S \}) \prod_{j \in S} \big( 1 - \alpha \beta^{x_j} \big)=\\
    &\sum_{\norm{\vec{x}}_1 = m} f(\vec{x}) g(\{ x_l \mid l \not\in S \}) \times\\
    &\times \sum_{j = 0}^{\abs{S}} (- \alpha)^j \sum_{S' \subseteq S: \abs{S'} = j} \prod_{k \in S'} \beta^{x_k},
    \label{eq:invariance-sum-step-1}
\end{split}
\end{equation}
where the inner most sum corresponds to summation over all subsets of $S$ with cardinality $j$.

Let us look at the class of bijections $\phi: S \to S$ and the corresponding bijection acting on the components of $\vec{x}$
\begin{equation}
    \varphi(x_l) = 
    \begin{cases}
        x_{\phi(l)} & \text{if } l \in S\\
        x_l         & \text{otherwise} 
    \end{cases}.
\end{equation}
For every subset $S'$ with cardinality $j$, there exists a map $\phi$ that maps $S' \mapsto S_o(j)$, where $S_o(j)$ is the set containing the $j$ smallest elements of $S$. Using this map to change variables in the sum in equation \eqref{eq:invariance-sum-step-1} from $\vec{x} \mapsto \varphi(\vec{x}) = \vec{y}$ yields
\begin{equation}
\begin{split}
    &I = \sum_{j = 0}^{\abs{S}} (- \alpha)^j \sum_{S' \subseteq S: \abs{S'} = j} \sum_{\norm{\varphi^{-1}(\vec{y})}_1 = m} f(\varphi^{-1}(\vec{y})) \times\\
    &\times g(\varphi^{-1}(\{ y_l \mid l \not\in S \})) \prod_{k \in S_o(j)} \beta^{y_k}.
\end{split}
\end{equation}
Using that $\varphi$ is the identity transformation on components with label $l \not\in S$, that $f$ is assumed to be invariant under the order of which the components are given and that a permutation preserves the $L_1$ norm yields
\begin{equation}
\begin{split}
    &I = \sum_{\norm{\vec{y}}_1 = m} f(\vec{y}) g(\{ y_l \mid l \not\in S \}) \sum_{j = 0}^{\abs{S}} (- \alpha)^j \times\\
    &\times \prod_{k \in S_o(j)} \beta^{y_k} \sum_{S' \subseteq S: \abs{S'} = j} 1.
\end{split}
\end{equation}
The number of subset of $S$ with cardinality $j$ is given by the number of ways to select $j$ elements out of $\abs{S}$. Hence the sum can be written
\begin{equation}
\begin{split}
    &I = \sum_{\norm{\vec{y}}_1 = m} f(\vec{y}) g(\{ y_l \mid l \not\in S \}) \sum_{j = 0}^{\abs{S}} (- \alpha)^j \binom{\abs{S}}{j} \prod_{k \in S_o(j)} \beta^{y_k}.
\end{split}
\end{equation}
\end{proof}

\begin{theorem}\label{the:photon-counting-distribution}
The photon counting distribution for a uniformly illuminated array consisting of $n$ indistinguishable single photon detectors with click probabilities given by equation \eqref{eq:single-photon-detector-probability} is given by
\begin{equation}
\begin{split}
    &\Pr(k \mid m; \eta, p_d, n) = \frac{1}{n^m} \binom{n}{k} \sum_{l = 0}^k (-1)^l \times\\
    &\times (1 - p_d)^{n - k + l} \binom{k}{l} \big[ n - (n - k + l) \eta \big]^m.
\end{split}
\end{equation}
\end{theorem}

\begin{proof}
Assume that $m$ photons are uniformly distributed over the array. The photon numbers on the single photon detectors are then multinomially distributed and the click probabilities are given by equation \eqref{eq:single-photon-detector-probability}. The photon counting distribution can therefore be written as
\begin{equation}
\begin{split}
    &\Pr(k \mid m; \eta, p_d, n) = \sum_{\norm{\vec{x}}_1 = m} \frac{m!}{\prod_{i = 1}^n x_i!} \frac{1}{n^m} \binom{n}{k} \times\\
    &\times \Bigg( \prod_{j = 1}^k \Pr(1 \mid x_j; \eta, p_d) \Bigg) \Bigg( \prod_{l = k + 1}^n \Pr(0 \mid x_l; \eta, p_d) \Bigg).
\end{split}
\end{equation}
Applying lemma \ref{the:sum-invariance} yields
\begin{equation}
\begin{split}
    &\Pr(k \mid m; \eta, p_d, n) = \frac{1}{n^m} \binom{n}{k} \sum_{j = 0}^k (-1)^j (1 - p_d)^{n - k + j}\times\\
    &\times \binom{k}{j} \sum_{\norm{\vec{x}}_1 = m} \frac{m!}{\prod_{i = 1}^n x_i!} \Bigg( \prod_{a = 1}^j (1 - \eta)^x_a \Bigg) \Bigg( \prod_{l = k + 1}^n (1 - \eta)^x_l \Bigg).
\end{split}
\end{equation}
The inner most sum can be transformed with the multinomial theorem giving
\begin{equation}
\begin{split}
    &\Pr(k \mid m; \eta, p_d, n) = \frac{1}{n^m} \binom{n}{k} \sum_{j = 0}^k (-1)^j \times\\
    &\times (1 - p_d)^{n - k + j} \binom{k}{j} \big[n - (n - k + j) \eta \big]^m.
\end{split}
\end{equation}
\end{proof}

\section{\label{sec:deriving-errors}ERROR ANALYSIS}
\begin{theorem}\label{the:binomial-convergence}
Assume that $p_d = 0$ then it holds that the photon counting distribution converges to a binomial distribution
\begin{equation}
    \Pr(k \mid m; \eta, 0, n) \to \binom{m}{k} \eta^k (1 - \eta)^{m - k},
\end{equation}
as $n \to \infty$.
\end{theorem}

\begin{proof}
Let us use the identity
\begin{equation}
    \binom{n}{k} = \frac{1}{k!} \sum_{j = 0}^k s(k, j) n^j,
\end{equation}
where $s(k, j)$ are the Stirling numbers of the first kind and the binomial formula on the photon counting distribution
\begin{equation}
\begin{split}
    &\Pr(k \mid m; \eta, 0, n) = \frac{1}{k!} \sum_{a = 0}^m \sum_{l, j = 0}^k \frac{(-1)^l s(k, j)}{n^{m - j - a}}\times\\
    &\times \binom{k}{l} \binom{m}{a} (1 - \eta)^a (k - l)^{m - a} \eta^{m - a}.
\end{split}
\end{equation}
Applying the identity
\begin{equation}
    S(m - a, k) = \frac{1}{k!} \sum_{l = 0}^k (-1)^l \binom{k}{l} (k - l)^{m - a},
\end{equation}
where $S(m - a, k)$ are the Stirling numbers of the second kind yields
\begin{equation}
\begin{split}
    &\Pr(k \mid m; \eta, 0, n) = \sum_{a = 0}^m \sum_{j = 0}^k \frac{s(k, j) S(m - a, k)}{n^{m - j - a}}\times\\
    &\times \binom{m}{a} (1 - \eta)^a \eta^{m - a}\\
    &= \sum_{j = 0}^k \sum_{b = j - m}^j s(k, j) S(j - b, k) \binom{m}{j - b} \eta^{j - b} \times\\
    &\times (1 - \eta)^{b + m - j} n^b.
    \label{eq:stirling-rewrite}
\end{split}
\end{equation}
We notice that $j - b < k$ if $b > 0$ which implies that $S(j - b, k) = 0$ when $b > 0$. Hence, the only term surviving in the limit $n \to \infty$ is when $b = 0$. We get
\begin{equation}
\begin{split}
    &\lim_{n \to \infty} \Pr(k \mid m; \eta, 0, n) = \sum_{j = 0}^k s(k, j) S(k, k) \binom{m}{j}\times\\
    &\times \eta^j (1 - \eta)^{m - j}.
\end{split}
\end{equation}
Using that $S(j, k) = 0$ if $j < k$ and that $s(k, k) S(k, k) = 1$ gives
\begin{equation}
    \lim_{n \to \infty} \Pr(k \mid m; \eta, 0, n) = \binom{m}{k} \eta^k (1 - \eta)^{m - k}.
\end{equation}
\end{proof}

%

\begin{theorem}\label{the:finite-size-error}
The finite size error when $\eta = 1$ is given by
\begin{equation}
\begin{split}
    \epsilon_n &= \frac{1}{2} \norm{\Pr(k \mid m; 1, 0, n) - \Pr(k \mid m; 1, 0, \infty)}_1\\
    &= \frac{1}{n} \binom{m}{2} + \order{n^{-2}}.
\end{split}
\end{equation}
\end{theorem}

\begin{proof}
Using the expression derived in equation \eqref{eq:stirling-rewrite} we get
\begin{equation}
\begin{split}
    \epsilon_n &= \frac{1}{2} \bigg\lVert \sum_{j = 0}^{\min(k, m - 1)} \sum_{b = j - m}^{-1} s(k, j) S(j - b, k) \times\\
    &\times \binom{m}{j - b} \eta^{j - b} (1 - \eta)^{m + b - j} n^b \bigg\rVert_1.
\end{split}
\end{equation}
The expression can be simplified by using that $(1 - \eta)^{m + b - j} = 0$ unless $m + b - j = 0$. This implies that
\begin{equation}
\begin{split}
    \epsilon_n &= \frac{1}{2} \bigg\lVert \sum_{j = 0}^{\min(k, m - 1)} s(k, j) S(m, k) n^{j - m} \bigg\rVert_1\\
    &= \frac{1}{2} \sum_{k = 0}^m \abs{\sum_{l = \max(1, m - k)}^m \frac{1}{n^l} s(k, m - l) S(m, k)}.
\end{split}
\end{equation}
To leading order in $n$ we get that
\begin{equation}
\begin{split}
    \epsilon_n &= \frac{1}{2n} \sum_{k = m - 1}^m \abs{s(k, m - 1) S(m, k)} + \order{n^{-2}}\\
    &= \frac{1}{n} \binom{m}{2} + \order{n^{-2}}.
\end{split}
\end{equation}
\end{proof}

%
\end{document}